\documentclass[runningheads]{llncs}
\usepackage[title]{appendix}
\usepackage{graphicx}
\usepackage[hidelinks]{hyperref}
\usepackage[all]{hypcap} % jump to the top of reference instead of bottom
\usepackage{algorithm}
\usepackage[noend]{algpseudocode}
\usepackage{multirow}
\usepackage{amsmath,amssymb}
\let\doendproof\endproof
\renewcommand\endproof{~\hfill\qed\doendproof}

\begin{document}
\title{Complexity of Public Goods Games on Graphs\thanks{This project has received funding from the European Research Council (ERC) under the European Union’s Horizon 2020 Research and Innovation Programme (grant agreement no. 740282).
A shorter version of this paper was first published in 15th Symposium on Algorithmic Game Theory (SAGT 2022) by Springer Nature.}}

\author{Matan Gilboa \and Noam Nisan}

\authorrunning{M. Gilboa and N. Nisan}

\institute{School of Computer Science and Engineering, Hebrew University of Jerusalem
\email{matan.gilboa@mail.huji.ac.il}\\
\email{noam@cs.huji.ac.il}}
\maketitle % typeset the header of the contribution
\begin{abstract}
We study the computational complexity of "public goods games on networks". In this model, each vertex in a graph is an agent that needs to take a binary decision of whether to "produce a good" or not. Each agent's utility depends on the number of its neighbors in the graph that produce the good, as well as on its own action. This dependence can be captured by a "pattern" $T:{\rm I\!N}\rightarrow\{0,1\}$ that describes an agent's best response to every possible number of neighbors that produce the good. Answering a question of [Papadimitriou and Peng, 2021], we prove that for some simple pattern $T$ the problem of determining whether a non-trivial pure Nash equilibrium exists is NP-complete. We extend our result to a wide class of such $T$, but also find a new polynomial time algorithm for some specific simple pattern $T$. We leave open the goal of characterizing the complexity for all patterns.

\end{abstract}

\section{Introduction}

We study scenarios where there is a set of agents, each of which must decide whether to make an effort to produce some ``good'', where doing so benefits not only himself but also others. This general type of phenomena is captured by the general notion of public goods, and we focus on cases where the public good produced by some agent does not benefit {\em all} other agents but rather there is some neighborhood structure between agents specifying which agents benefit from the public goods produced by others. Examples for these types of scenarios abound: anti-pollution efforts that benefit a geographical neighborhood, research efforts that benefit other researchers in related areas, vaccination efforts, security efforts, and many more.

We focus on the following standard modeling of {\em public goods games on networks}: We are given an undirected graph, where each node models an agent, and the neighbors of a node are the other agents that benefit from his production of the public good. We focus on the case where each agent has a single boolean decision to make of whether to produce the good or not, and furthermore, as in, e.g., \cite{Modifications,Parameterized,Directed_Paper,Point_Out_Flaw,Corrected}, limit ourselves to cases where the effect of an agent's neighbors is completely characterized by the number of them that produce the good. Furthermore we focus on the cleanest, so called, {\em fully homogenous} case where all agents have the same costs and utility functions so all of the heterogeneity between agents is captured by the graph structure. 

Formally, there is a cost $c$ that each agent pays if they produce their good, and a utility function $u(s_i,k_i)$ describing each agent's utility, where $s_i \in \{0,1\}$ describes whether agent $i$ produces the good and $k_i \in {\rm I\!N}$ is the number of $i$'s neighbors (excluding $i$) that produce the good\footnote{In this paper, we use the notation ${\rm I\!N}$ to also contain 0, and use ${\rm I\!N}_+$ to exclude 0.}. We focus on the 'strict' version of the problem \cite{Parameterized}, where we do not allow knife's-edge cases where $u(1,k_i)-u(0,k_i)=c$. Therefore, our agent's best response to exactly $k_i$ of its neighbors producing the good is to produce the good ($s_i=1$) if $u(1,k_i)-u(0,k_i)>c$ and to not produce the good ($s_i=0$) if $u(1,k_i)-u(0,k_i)<c$. Hence, we can summarize the Best Response Pattern as $T:{\rm I\!N} \rightarrow \{0,1\}$. We study the following basic problem of finding a non-trivial pure Nash equilibrium in a network.

\vspace{0.1in}
\noindent
{\bf Equilibrium in a public goods game:} For a given best response pattern $T:{\rm I\!N} \rightarrow \{0,1\}$, given as input an undirected graph $G=(V,E)$, determine whether there exists a pure non-trivial Nash equilibrium of the public goods game on $G$, i.e. an assignment $s:V \rightarrow \{0,1\}$ that is not all $0$, such that for every $1\leq i \leq |V|$ we have that $s_i=T(\sum_{(i,j)\in E} s_j)$.
\vspace{0.1in}

Several cases of this problem have been studied in the literature. In \cite{IS_2007}, the ``convex case'' where the pattern is monotone (best response is $1$ if at least $t$ of your neighbors play $1$) and the "Best Shot" case (where the best response is $1$ only if none of your neighbors play $1$, i.e. $T(0)=1$ and for all $k>0$, $T(k)=0$) were shown to have polynomial time algorithms. The general heterogeneous case where different agents may have different patterns of best responses was shown to be NP-complete in \cite{Corrected} as was, in \cite{Parameterized,Point_Out_Flaw}, the fully-homogenous case, if we allow also knife's-edge cases, i.e. for some utility function where for some $k$ we have $u(1,k)-u(0,k)=c$ (which lies outside our concise formalization of patterns, since in these cases, both $1$ and $0$ are best responses).\footnote{An early version \cite{Flawed} of \cite{Corrected} contained an erroneous proof of NP-completeness in the fully-homogenous case for some pattern $T$, but a bug was found by \cite{Point_Out_Flaw} who gave an alternative proof of the NP-completeness of the case that allows $u(1,k)-u(0,k)=c$.} The parameterized complexity for several natural parameters of the graph was studied in \cite{Parameterized}. In \cite{Modifications}, it is shown that in a public goods game, computing an equilibrium where at least $k$ agents produce the good, or at least some specific subset of agents produce it, is NP-Complete. In \cite{Directed_Paper}, a version of this problem\footnote{A version without the non-triviality assumption on the equilibrium.} on {\em directed} graphs was studied, and a full characterization of the complexity of this problem was given for every pattern: except for a few explicitly given best response patterns, for any other pattern the directed problem is NP-complete.\footnote{One of the few easy cases they identify can be seen to apply also to the undirected case: where the pattern alternates between 0 and 1, i.e. $T(k)$ only depends on the parity of $k$, a case that can be solved as a solution of linear equations over $GF(2)$.} They also suggested an open problem of providing a similar characterization for the more standard undirected case and specifically asked about the complexity of the pattern where the best response is $1$ iff exactly one of your neighbors plays $1$. Our main result answers this question, showing that for this specific pattern the problem is NP-complete.

\vspace{0.1in}
\noindent
{\bf Theorem:} For the Best-Response Pattern where each agent prefers to produce the good iff exactly one of its neighbors produces the good, i.e. $T(1)=1$ and for all $k \ne 1$ $T(k)=0$, the equilibrium decision problem in a public goods game is NP-complete.
\vspace{0.1in}

When considering the strict version of the game, this is the first pattern for which the equilibrium problem is shown to be NP-complete, and in fact it is even the first proof that the general problem (where the pattern is part of the input) is NP-complete. We then embark on the road to characterizing the complexity for all possible patterns. We extend our proof of NP-completeness to large classes of patterns. We also find a new polynomial time algorithm for a new interesting case:

\vspace{0.1in}
\noindent
{\bf Theorem:} For the pattern where each agent prefers to produce the good iff at most one of its neighbors produces the good, i.e. $T(0)=T(1)=1$ and for all $k > 1$, $T(k)=0$, the public goods game always has a pure non-trivial equilibrium, and it can be found in polynomial time.
\vspace{0.1in}

We were not able to complete our characterization for all patterns and leave this as an open problem. In particular, we were not able to determine the complexity for the following two cases:

\vspace{0.1in}
\noindent
{\bf Open Problem 1:} Determine the computational complexity of the equilibrium decision problem of a public goods game for the pattern where $T(0)=T(1)=T(2)=1$ and for all $k > 2$, $T(k)=0$.
\vspace{0.1in}

\vspace{0.1in}
\noindent
{\bf Open Problem 2:} Determine the computational complexity of the equilibrium decision problem of a public goods game for the pattern where $T(0)=T(2)=1$ and for all $k \notin \{0,2\}$, $T(k)=0$.
\vspace{0.1in}

\noindent We suspect that at least the first problem is computationally easy, and in fact that there exists a non-trivial pure Nash equilibrium in any graph.

\begin{table}[h!]
\begin{center}
\begin{tabular}{ |c|c|c| } 
 \hline
 \textbf{Category} & \textbf{Pattern} & \textbf{Reference} \\
 \hline
 \multirow{4}{4em}{PTIME} & 1,0,0,0,... & \cite{IS_2007} \\ 
 & 1,1,0,0,0,... & Theorem \ref{Theorem_AMSN_DP} \\ 
 & 0,0,...,1,1,1,... & \cite{IS_2007} \\ 
 & 1,0,1,0,1,0,... & \cite{Directed_Paper} \\ 
 \hline
 \multirow{5}{4em}{NPC} & 0,1,0,0,0,... & Theorem \ref{theorem_SN_npc} \\ 
 & 0,?,?,...,1,0,0,0,... & Theorem \ref{theorem_flat_NPC} \\ 
 & 1,1,?,?,...,0,?,?,...,1,0,0,0,... & Theorem \ref{theorem_sloped_NPC} \\ 
 & 1,0,..,0,1,1,?,?,...,0,0,0,... & Theorem \ref{theorem_+-++} \\ 
 & Adding 1,0 to non-flat hard patterns & Theorem \ref{theorem_add_+-_NPC} \\ 
 \hline
 \multirow{2}{4em}{Open Problems} & 1,1,1,0,0,0,... & \\ 
 & 1,0,1,0,0,0,... & \\ 
 \hline
\end{tabular}
\caption{\label{results_table}A summary of our (and previous) results.}
\end{center}
\end{table}

The rest of this paper is organized as follows: After defining our model and notations in section \ref{section_model}, we present our main theorem (hardness of the Single-Neighbor pattern) in Section \ref{section_SN}, and provide some intuition about the problem. In Section \ref{section_AMSN} we construct a polynomial time algorithm for the At-Most-Single-Neighbor pattern. In Section \ref{section_more_hard_patterns} we characterize a number of classes of patterns for which the problem is hard, by reducing from our main theorem, where each sub-section focuses on a specific class of patterns. Our results are summarized in Table \ref{results_table}.

\section{Model and Notation}
\label{section_model}
We define a \textit{Public Goods Game} (PGG) on an undirected graph $G=(V,E)$ with $n$ nodes $V=\{1,...,n\}$, each one representing an agent. The \textit{neighborhood} of agent \textit{i}, denoted $N(i)$, is defined as the set of agents adjacent to $i$, excluding $i$, i.e: $N(i)=\{j|(j,i)\in E\}$. An edge between two agents' nodes models the fact that these agents are directly affected by each other's decision to produce or not produce the good. The strategy space, which is assumed to be the same for all agents, is $S=\{0,1\}$, where $1$ represents producing the good, and $0$ represents the opposite. The strategy of agent $i$ is denoted $s_i\in S$. 
\begin{definition}
If nodes $i,j$ are adjacent, and $s_j=1$ (i.e. agent $j$ produces the good) we say that $j$ is a \textbf{supporting neighbor} of $i$.
\end{definition}
For convenience, if some node $v$ represents agent $i$, we sometimes write $v=1$ or $v=0$ instead of $s_i=1$ or $s_i=0$ respectively, to mark $i$'s strategy. The utility function is assumed to be the same for all agents. Furthermore, we restrict ourselves to utility functions where an agent is never indifferent between producing and not producing the good, and so always has a single best response according to the strategies of the agents in their neighborhood. This characteristic of the utility function allows us to adopt a more convenient way to inspect a PGG model, which we call the \textit{best response pattern}.
\begin{definition}
For any PGG, we define its Best Response Pattern (BRP), denoted by $T$, as an infinite boolean vector in which the $k^{th}$ entry represents the best response for each agent $i$ given that exactly $k$ neighbors of $i$ (excluding $i$) produce the good:
\begin{align*}
\forall k\in {\rm I\!N} \;\; T[k]= \text{best response to $k$ productive neighbors}
\end{align*}
\end{definition}
We henceforth identify PGGs by their Best Response Pattern, rather than their utility function and cost. This concludes the definition of a PGG model. We now define a \textit{pure Nash equilibrium}, which is our main subject of interest.

\begin{definition}
A strategy profile $s=(s_1,...,s_n)\in S^n$ of a Public Goods Game corresponding to a BRP $T$ is a pure Nash equilibrium (PNE) if all agents play the best response to the strategies of the agents in their neighborhood:
\begin{align*}
\forall i\in [n]\;\; s_i=T[\sum_{j\in N(i)} s_j]
\end{align*}
In addition, if there exists $i\in [n]$ s.t $s_i=1$, then $s$ is called a non-trivial pure Nash equilibrium (NTPNE).
\end{definition}

\begin{definition}
For a fixed BRP $T$, the non-trivial pure Nash equilibrium decision problem corresponding to $T$, denoted by NTPNE($T$), is defined as follows:
The input is an undirected graph $G$. The output is 'True' if there exists an NTPNE in the PGG defined on $G$ with respect to $T$, and 'False' otherwise.
The search version of the problem asks for the NTPNE itself.
\end{definition}

Let us give names to the following three simplest patterns:
\begin{definition}
The Best-Shot best-response pattern is defined as follows:
\begin{align*}
\forall k\in {\rm I\!N} \;\; T[k]=    
    \begin{cases}
      1 & \text{if $k=0$}\\
      0  & \text{if $k\geq1$}
    \end{cases}   
\end{align*}
i.e. 
\begin{align*}
T=[1,0,0,0,0,...]
\end{align*}
\end{definition}

\begin{definition}
The Single-Neighbor best-response pattern is defined as follows:
\begin{align*}
\forall k\in {\rm I\!N} \;\; T[k]=    
    \begin{cases}
      1 & \text{if $k=1$}\\
      0  & \text{otherwise}
    \end{cases}
\end{align*}
i.e. 
\begin{align*}
T=[0,1,0,0,0,...]
\end{align*}
\label{def_SN_BRP}
\end{definition}

\begin{definition}
The At-Most-Single-Neighbor best-response pattern is defined as follows:
\begin{align*}
\forall k\in {\rm I\!N} \;\; T[k]=    
    \begin{cases}
      1 & \text{if $k\leq 1$}\\
      0  & \text{if $k>1$}
    \end{cases}   
\end{align*}
i.e.
\begin{align*}
T=[1,1,0,0,0,...]
\end{align*}
\label{AMSN_BRP}
\end{definition}

The Best-Shot BRP was coined in \cite{IS_2007}, where they prove that a pure Nash equilibrium exists in any graph, and show a correspondence between PNEs and Maximal Independent Sets. We study the Single-Neighbor BRP in Section \ref{section_SN}, where we prove the decision problem is NP-Complete. We study the At-Most-Single-Neighbor BRP in Section \ref{section_AMSN}, where we prove that a pure Nash equilibrium exists in any graph.
\section{Hardness of the Single-Neighbor Pattern}
\label{section_SN}
In this section we prove NP-completeness of NTPNE($T$) defined by the Single-Neighbor BRP, and provide basic intuition about its combinatorial structure. This is a linchpin of our hardness results, from which we reduce to many other patterns.
We remind the reader that in the Single-Neighbor BRP, an agent prefers to produce the good iff exactly one of their neighbors produces it.

\begin{theorem}
Let $T$ be the Single-Neighbor Best Response Pattern. Then NTPNE($T$) is NP-complete.
\label{theorem_SN_npc}
\end{theorem}

Before the proof, we provide intuition about the Single-Neighbor problem, by examining a few simple graphs. First, we note that since $T[0]=0$, a trivial all-zeros PNE exists in any graph. This observation is true for any such pattern\footnote{These patterns are denoted \textit{flat} patterns, and are formally defined in Definition \ref{def_flat}.}, which is the reason we choose to focus on non-trivial PNEs. Now, take for example a simple path with two nodes. The assignment where both nodes are set to $1$ is an NTPNE, since neither of them benefit from changing their strategy. But looking at a simple path with 3 nodes, it is easy to verify that there is no NTPNE. Specifically, the all-ones assignment in such a path is \textit{not} a PNE since the middle node would rather play $0$, as it already has two supporting neighbors.
Generalizing this NTPNE analysis to paths of any size, we see that in order for a simple path with $n$ nodes $x_1,...,x_n$ to have an NTPNE, it must be that $n\equiv2\pmod{3}$. To see why, let us examine $x_1$. If $x_1$ is assigned $0$ then so is $x_2$, as otherwise $x_1$ wishes to change strategy; and since $x_2$ is assigned $0$ then so is $x_3$, and so forth. Therefore, in order to get a non-trivial assignment $x_1$ must be assigned $1$, and it must have a supporting neighbor, which must be $x_2$. This leads to only one possibility for an NTPNE, as shown in Figure \ref{Intuition1_path_n}. A similar analysis shows that a cycle with $n$ nodes has an NTPNE iff $n\equiv0\pmod{3}$ (see Figure \ref{Intuition2_cycle_n}).

\begin{figure}[h!]
\centering
\includegraphics[width=0.6\textwidth]{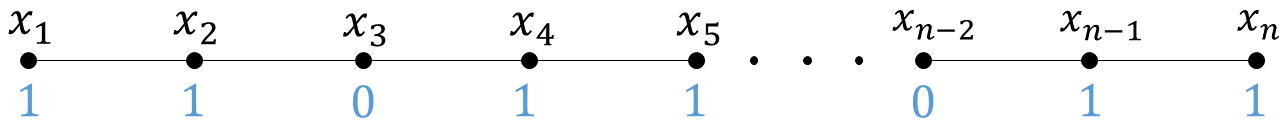}
\caption{\textit{NTPNE: paths with $n\equiv2\pmod{3}$}}
\label{Intuition1_path_n}
\end{figure}
Another simple example is the Complete Graph, or Clique. In any Clique of size at least 2, we can construct an NTPNE by choosing any two nodes to be assigned $1$, and assigning $0$ to all other nodes (see Figure \ref{Intuition3_clique}).
\begin{figure}[h!]
\centering
\begin{minipage}[t]{.35\textwidth}
    \includegraphics[width=\textwidth]{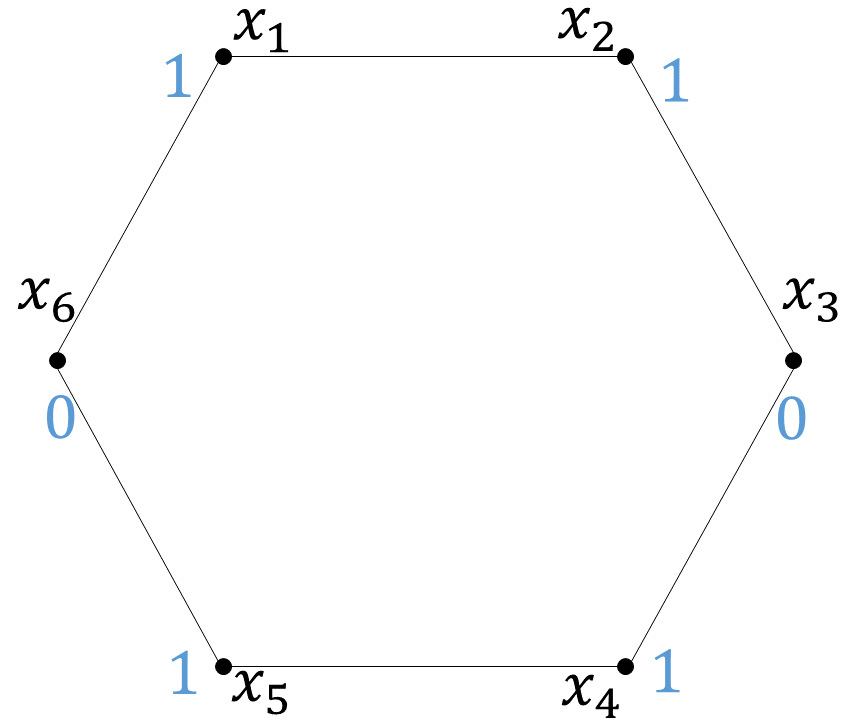}
    \caption{\textit{NTPNE: cycles with $n\equiv0\pmod{3}$}}
    \label{Intuition2_cycle_n}
\end{minipage}%
\hspace{10mm}
\begin{minipage}[t]{.35\textwidth}
    \includegraphics[width=\textwidth]{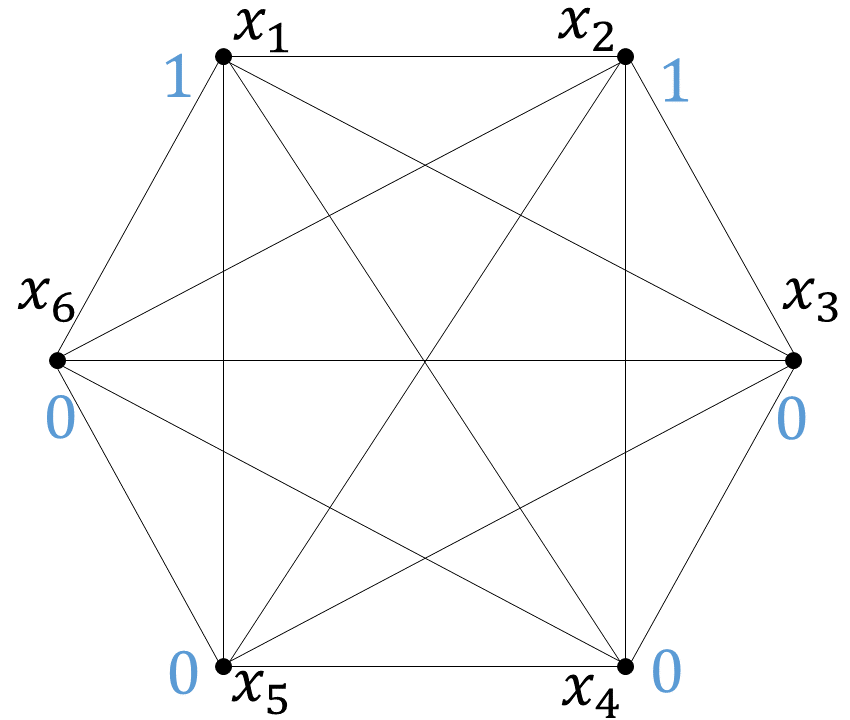}
    \caption{\textit{NTPNE: Clique}}
    \label{Intuition3_clique}
\end{minipage}
\end{figure}
So we see there are cases where there exists an NTPNE, and others where there doesn't, and so the problem is not trivial (and in fact is NP-Hard).

\bigskip
We now begin the proof of Theorem \ref{theorem_SN_npc}, first showing that the problem is NP-Hard. To do so, we construct a reduction from \textit{ONE-IN-THREE 3SAT}, which is a well known NP-complete problem \cite{one_in_three_3sat}. The input of the \textit{ONE-IN-THREE 3SAT} problem is a CNF formula where each clause has exactly 3 literals, and the goal is to determine whether there exists a boolean assignment to the variables such that in each clause exactly one literal is assigned with \textit{True}.

For the reduction, we introduce our Clause Gadget. For each clause $(l_1,l_2,l_3)$ in the \textit{ONE-IN-THREE 3SAT} instance, we construct a 9-nodes Clause Gadget as demonstrated in Figure \ref{A2_cg_ne}. The nodes $l_1,l_2,l_3$ represent the literals of the clause, respectively, and are denoted the \textit{Literal Nodes}. Nodes $a,b,c$ are denoted the \textit{Inner Nodes}, and nodes $x,y,z$ are denoted the \textit{Peripheral Nodes}. Each Literal Node is adjacent to all other Literal Nodes, and to all Peripheral Nodes. In addition, the Literal Nodes $l_1,l_2,l_3$ are \textit{paired} with the Peripheral Nodes $x,y,z$ respectively, in the sense that they share the same Inner Node as a neighbor, and only that Inner Node (for example, $l_1$ and $x$ are \textit{paired} since they are both adjacent to $a$, and not to $b,c$). Notice that the Peripheral Nodes are \textit{not} adjacent to each other. Additionally, note that in the final graph, only the Literal Nodes will be connected to nodes outside the Clause Gadget. The proof is constructed by a number of Lemmas.
\begin{figure}[h!]
\centering
\begin{minipage}[t]{.4\textwidth}
    \includegraphics[width=\textwidth]{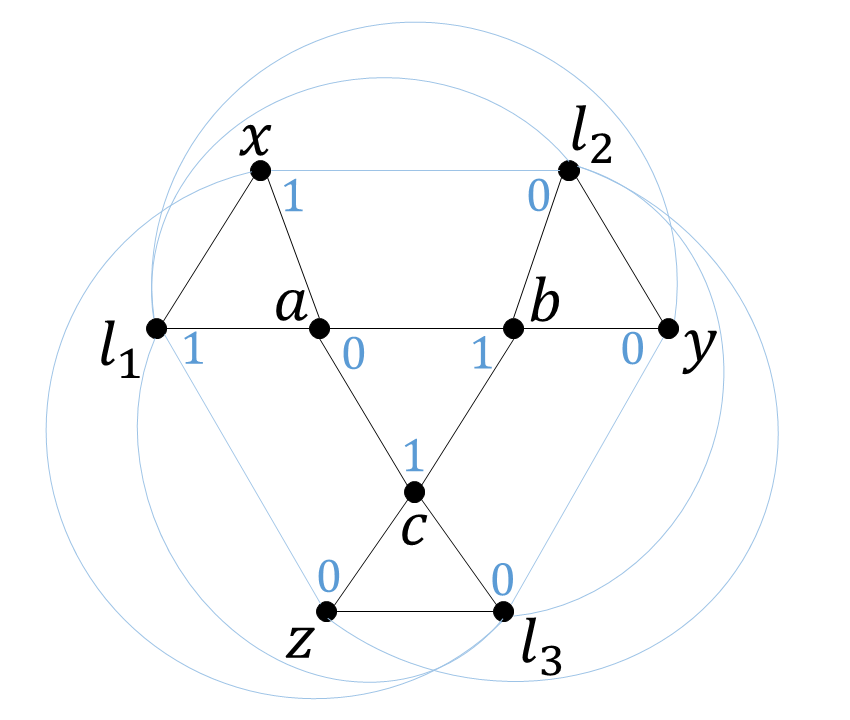}
    \caption{\textit{Clause Gadget, with the NTPNE assignment of Lemma \ref{lemma_cg_pne}}}
    \label{A2_cg_ne}
\end{minipage}%
\begin{minipage}[t]{.6\textwidth}
    \includegraphics[width=\textwidth]{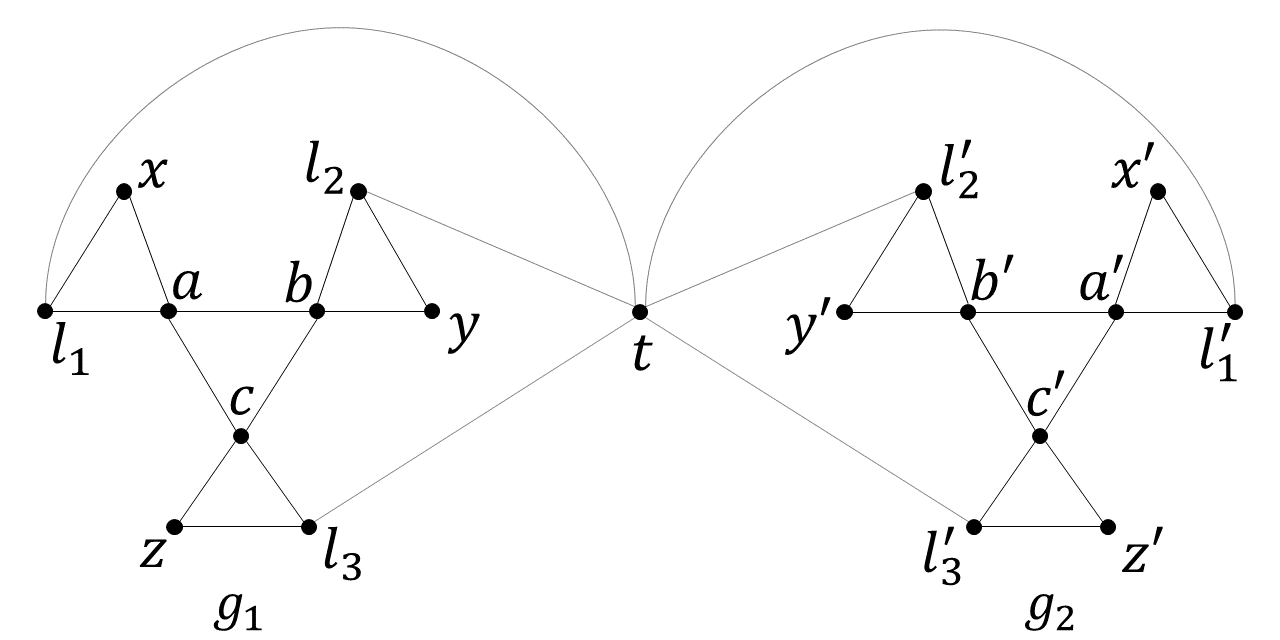}
    \caption{\textit{Transfer Node (blue edges not shown)}.}
    \label{A3_transfer}
\end{minipage}
\end{figure}
\begin{lemma}
In any PNE in a graph which includes the Clause Gadget, if one of the nodes of the Clause Gadget is assigned $1$, then one of the Literal Nodes must be assigned $1$.
\label{lemma_literal_on}
\end{lemma}

\begin{proof}
Assume by way of contradiction that all Literal Nodes are assigned $0$. Since at least one node in the gadget is assigned $1$, it must either be an Inner Node or a Peripheral Node. Notice that if an Inner Node is assigned $1$, w.l.o.g $a=1$, then its neighboring Peripheral Node $x$ has exactly one supporting neighbor and according to the BRP must also be assigned $1$, seeing that x is connected only to $a$ and to the Literal Nodes (which are assigned $0$). Similarly, if a Peripheral Node is assigned $1$, w.l.o.g $x=1$, then its neighboring Inner Node $a$ must also be assigned $1$, as otherwise $x$ would prefer changing strategy.
Therefore, there must be a pair of adjacent Inner Node and Peripheral Node that are both assigned $1$. w.l.o.g $a=x=1$. Since $a$ is assigned $1$, and already has a supporting neighbor, then all other neighbors of $a$ (i.e. $b$ and $c$) must be set to $0$. This leaves us only with $z,y$; since neither of them have any supporting neighbors, they must be set to $0$. The contradiction comes from nodes $b,c$, both of which prefer changing their strategy to $1$, having exactly one supporting neighbor. 
\end{proof}

\begin{lemma} In any PNE in a graph which includes the Clause Gadget, if one of the Literal Nodes of the Clause Gadget is assigned $1$, then the other two Literal Nodes must be assigned $0$.
\label{lemma_one_literal}
\end{lemma}

\begin{proof} Assume by way of contradiction that two different Literal Nodes are assigned $1$ (w.l.o.g $l_1=l_2=1$). Since $l_1$ and $l_2$ are adjacent, they both already have a supporting neighbor, and so all their other neighbors must be set to $0$. Therefore $l_3=x=y=z=a=b=0$. This leaves us only with node $c$, which must be set to $0$ since all its neighbors are set to $0$. The contradiction comes from nodes $a,b$, both of which prefer changing their strategy to $1$, having exactly one supporting neighbor. 
\end{proof}

\begin{lemma} In any PNE in a graph which includes the Clause Gadget, if one of the Literal Nodes of the Clause Gadget is assigned $1$, then so is its paired Peripheral Node.
\label{lemma_literal_peripheral}
\end{lemma}

\begin{proof} Assume by way of contradiction that a Literal Node is set to $1$ while its \textit{paired Peripheral Node} is set to $0$. w.l.o.g $l_1=1, x=0$. From Lemma \ref{lemma_one_literal}, we have that $l_2=l_3=0$. Therefore, since $x$ cannot have only one supporting neighbor and still prefer playing $0$, we must set its remaining neighbor, $a$, to $1$. Since $l_1,a$ both have a supporting neighbor, all their other neighbors must be set to $0$. Therefore $b=c=y=z=0$. The contradiction comes from nodes $b,c$, both of which prefer changing their strategy to $1$, having exactly one supporting neighbor. 
\end{proof}

\begin{lemma} In any PNE in a graph which includes a Clause Gadget, if one of the nodes of the Clause Gadget is assigned $1$, then there is only one possible assignment to the nodes in the gadget. Specifically, one Literal Node and its paired Peripheral Node must be set to $1$, and so do the two Inner Nodes that aren't connected to them, whereas all other nodes in the gadget must be set to $0$.
\label{lemma_cg_pne}
\end{lemma}

\begin{proof} Since there exists a node that is set to $1$ inside the gadget, From Lemma \ref{lemma_literal_on} one of the Literal Nodes must be set to $1$, w.l.o.g $l_1=1$. From Lemma \ref{lemma_one_literal} $l_2=l_3=0$, and from Lemma \ref{lemma_literal_peripheral} $x=1$. Since $l_1,x$ are supporting neighbors to each other, they cannot have any other neighbor set to $1$, therefore $a=y=z=0$. Since $y$ is set to $0$ and has only one supporting neighbor ($l_1$), we must set its remaining neighbor $b$ to $1$ as well. Symmetrically, we must set $c$ to $1$ in order to support $z$'s assignment. It is easy to verify that indeed each node of the Clause Gadget is playing its best response given this assignment. 
\end{proof}
So far we have seen that the Clause Gadget indeed permits an NTPNE, and enforces the fact that each clause of the CNF formula must have exactly one literal set to $1$. We now wish to enforce the fact that \textit{all} clauses must have a literal assigned with \textit{True}. We first construct a connection between the Clause Gadgets such that \textit{if} some Clause Gadget has a node set to $1$, then all gadgets must have one. The connection is defined as follows. Each pair\footnote{It is enough to connect all Clause Gadgets as a chain to one another (by Transfer Nodes), but for ease of proof we connect every pair of gadgets.\label{chain_cg_footnote}}
of Clause Gadgets is connected by one \textit{Transfer Node}, denoted by $t$. The Transfer Node is adjacent to all Literal Nodes of both of the gadget to which it is connected, and only to those nodes. The connection between the Clause Gadgets is demonstrated in Figure \ref{A3_transfer}.

\begin{lemma} In any PNE in a graph which includes 2 Clause Gadgets which are connected by a Transfer Node $t$, $t$ must be set to $0$.
\label{lemma_transfer_off}
\end{lemma}

\begin{proof} Assume by way of contradiction that $t=1$. Then it must have a supporting neighbor. Since $t$ is only connected to Literal Nodes, one of those Literal Nodes must be set to $1$ (w.l.o.g $l_1=1$). From Lemma \ref{lemma_literal_peripheral}, $x$ must also be assigned with $1$, which leads to a contradiction since $l_1$ has 2 supporting neighbors and yet plays $1$.  
\end{proof}

\begin{lemma} In any PNE in a graph which includes at least two Clause Gadgets, which are all connected to each other by Transfer Nodes, if one of the Clause Gadgets has a node set to $1$, then all of the Clause Gadgets have a node set to $1$.
\label{lemma_transfer_works}
\end{lemma}

\begin{proof} Denote the gadget that has a node set to $1$ by $g_1$, and let $g_2$ be some other Clause Gadget. Then $g_2$ is connected to $g_1$ via a Transfer Node $t$. Denote the Literal Nodes in $g_1$ by $l_1,l_2,l_3$, and the Literal Nodes in $g_2$ by $l'_1,l'_2,l'_3$. From Lemma \ref{lemma_literal_on} and Lemma \ref{lemma_one_literal} we have that one of $l_1,l_2,l_3$ is set to $1$, while the other two are set to $0$. w.l.o.g assume $l_1=1, l_2=l_3=0$. From Lemma \ref{lemma_transfer_off}, $t=0$, and therefore, $t$ must have another supporting neighbor other than $l_1$. Since $t$'s only neighbors are $l_1,l_2,l_3,l'_1,l'_2,l'_3$, and $l_2=l_3=0$, it follows that one of $l'_1,l'_2,l'_3$ must be set to one, while the other two must be set to $0$. From Lemma \ref{lemma_cg_pne} we know the assignments in each Clause Gadget necessary for an NTPNE, and it is easy to verify that this assignment is still a Nash Equilibrium after adding Transfer Nodes between the gadgets. 
\end{proof}

Now that we have ensured all Clause Gadgets have a node set to $1$ (assuming one of them does), we wish to enforce that any two identical literals in the CNF formula are assigned with the same value. To do so, we introduce another connecting node, which we call the \textit{Copy Node}. Any two\footnote{It is enough to connect all literals representing the same variable as a chain to one another (by Copy Nodes), but for ease of proof we connect every pair of them.\label{chain_copy_footnote}} Literal Nodes $l_1, l'_1$ (from different Clause Gadgets, or possibly from the same one) which represent the same variable in the original CNF formula, will be connected via a Copy Node denoted by $k$, as shown in Figure \ref{A4_copy}. Each Copy Node has exactly two neighbors, which are $l_1,l'_1$.
\begin{figure}[h!]
\centering
\includegraphics[width=0.7\textwidth]{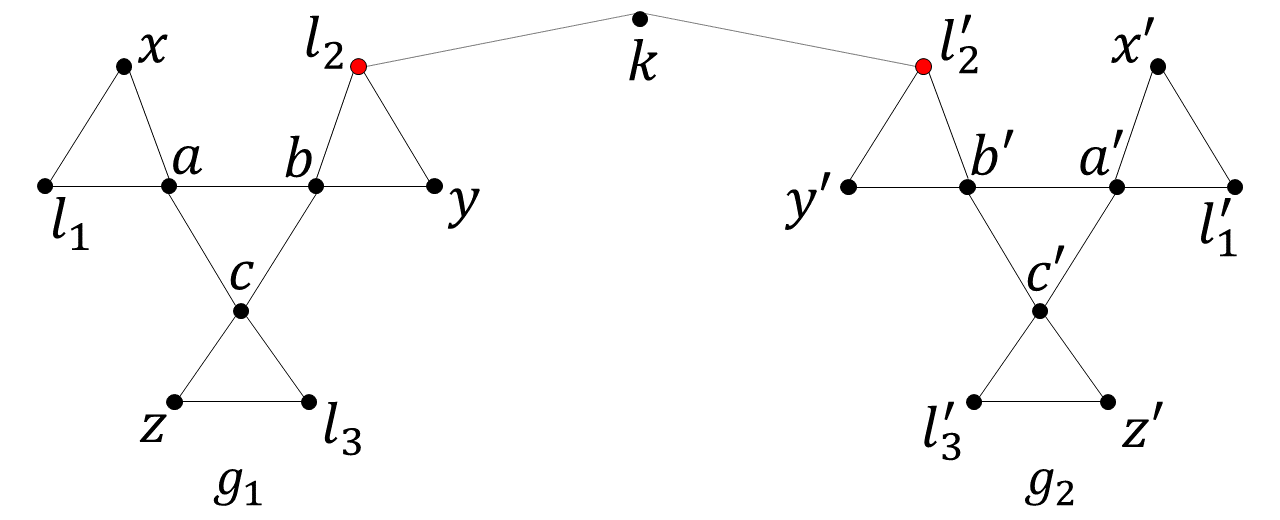}
\caption{\textit{Copy Node (blue edges not shown)}.}
\label{A4_copy}
\end{figure}
\begin{lemma} In any PNE in a graph which includes two Literal Nodes $l_1,l'_1$ in two Clause Gadgets $g_1,g_2$ respectively, where $l_1,l'_1$ are connected by a Copy Node $k$, $k$ must be set to $0$.
\label{lemma_copy_off}
\end{lemma}
\begin{proof}
Since $k$ is connected only to Literal Nodes, the proof of Lemma \ref{lemma_transfer_off} applies to this claim as well. 
\end{proof}
\begin{lemma} In any PNE in a graph which includes two Literal Nodes $l_1,l'_1$ in two Clause Gadgets $g_1,g_2$ respectively, where $l_1,l'_1$ are connected by a Copy Node $k$, $l_1$ and $l'_1$ must have the same assignment.
\label{lemma_copy_works}
\end{lemma}
\begin{proof} Assume by way of contradiction that (w.l.o.g) $l_1=1, l'_1=0$. From Lemma \ref{lemma_copy_off}, $k$ must be set to $0$. But since $l_1,l'_1$ are the only neighbors of $k$, $k$ has only one supporting neighbor, and therefore $k$ must be set to $1$, in contradiction. 
\end{proof}

The next property of a \textit{ONE-IN-THREE 3SAT} assignment we need to enforce, is that a variable $x$ and its negation $\overline{x}$ must be set to different values. Since the Copy Nodes already ensure that a variable appearing several times will always get the same value, it is enough to make sure for each variable that one instance of it is indeed different from one instance of its negation.
To do so, we introduce another connecting node, called the \textit{Negation Node}. For each variable $x$, where both $x$ and $\overline{x}$ appear in the CNF formula, we choose one instance of $x$ and one instance of $\overline{x}$ from different clauses.\footnote{We assume a variable and its negation never appear together in the same clause, as the problem without this assumption is easily reducible to the problem with it.} Denote the Literal Nodes representing $x,\overline{x}$ by $l_1,l'_1$ respectively, and denote the other two Literal Nodes residing with $l'_1$ in the same Clause Gadget by $l'_2,l'_3$. We connect a Negation Node $n$ to $l_1$ as well as to $l'_2,l'_3$, as demonstrated in Figure \ref{A5_neg}. For convenience, we say that $l_1,l'_1$ are connected by $n$ even though $n$ is only adjacent to one of them.
\begin{figure}[h!]
\centering
\includegraphics[width=0.7\textwidth]{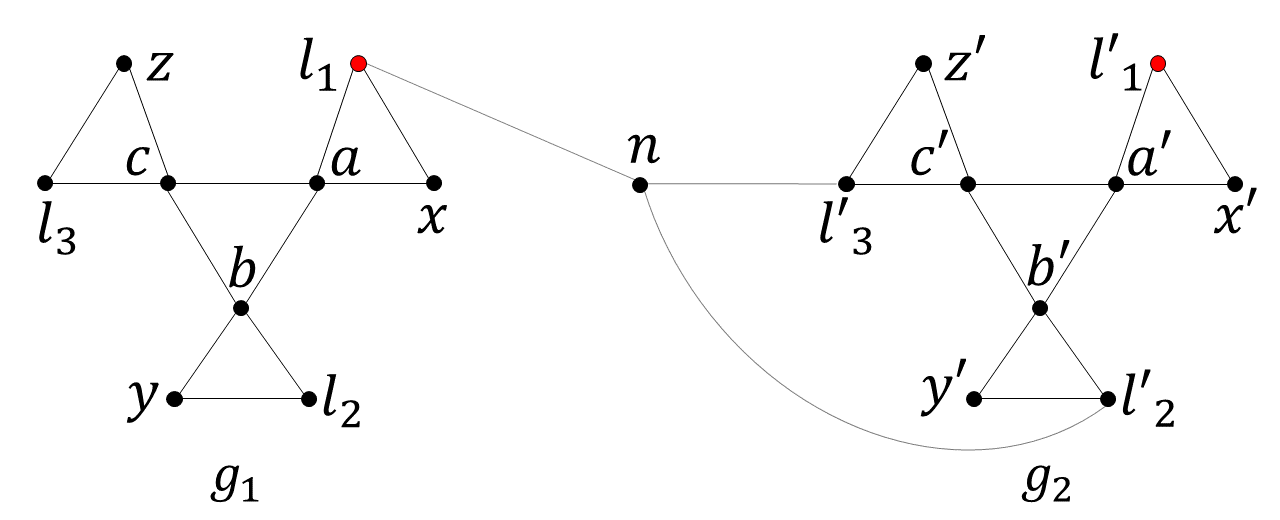}
\caption{\textit{Negation Node (blue edges not shown)}.}
\label{A5_neg}
\end{figure}
\begin{lemma} In any PNE in a graph which includes two Literal Nodes $l_1,l'_1$ in two different Clause Gadgets $g_1,g_2$ respectively, where $l_1,l'_1$ are connected by a Negation Node $n$, then $n$ must be set to $0$.
\label{lemma_neg_off}
\end{lemma}
\begin{proof} 
Since $n$ is connected only to Literal Nodes, the proof of Lemma \ref{lemma_transfer_off} applies to this claim as well. 
\end{proof}
\begin{lemma} In any PNE in a graph which includes two Literal Nodes $l_1,l'_1$ in two different Clause Gadgets $g_1,g_2$ respectively, where $l_1,l'_1$ are connected by a Negation Node $n$, and each of $g_1,g_2$ has at least one node set to $1$, then $l_1,l'_1$ must be assigned with different values.
\label{lemma_neg_works}
\end{lemma}
\begin{proof} Denote by $l'_2,l'_3$ the other two Literal Nodes residing with $l'_1$ in the same Clause Gadget. Divide into cases.
Case 1: If $l_1=1$, then $n$ must have another supporting neighbor (since $n$ itself is set to $0$, from Lemma \ref{lemma_neg_off}). Thus, either $l'_2$ or $l'_3$ must be set to $1$, and from Lemma \ref{lemma_one_literal} we have that $l'_1=0$, as needed.
Case 2: If $l_1=0$, assume by way of contradiction $l'_1=0$. Since $g_2$ has some node set to $1$, then from Lemmas \ref{lemma_literal_on} and \ref{lemma_one_literal} we have that either $l'_2$ or $l'_3$ are set to $1$, and only one of them. Thus, $n$ has exactly one supporting neighbor, and therefore $n$ must be set to $1$, in contradiction to Lemma \ref{lemma_neg_off}. 
\end{proof}

We have shown that our construction enforces all properties of a valid \textit{ONE-IN-THREE 3SAT} solution. But notice that most of the proofs rely on the assumption that every Clause Gadget has at least one node set to $1$. It is left to prove that there cannot be any NTPNE where all of the Clause Gadgets are all-zero.
\begin{lemma} In any NTPNE in the graph constructed throughout the proof of Theorem \ref{theorem_SN_npc}, all Clause Gadgets have at least one node set to $1$.
\label{lemma_all_cg_on}
\end{lemma}
\begin{proof} Since we have an NTPNE, there must be some non-zero node, denoted by $v$. From Lemmas \ref{lemma_transfer_off}, \ref{lemma_copy_off}, and \ref{lemma_neg_off} we have that all Transfer Nodes, Copy Nodes and Negation Nodes must be set to $0$, and therefore $v$ must be a part of a Clause Gadget. \footnote{Lemmas \ref{lemma_transfer_off}, \ref{lemma_copy_off}, and \ref{lemma_neg_off} do not assume that the Clause Gadgets have a node set to $1$, and therefore are valid to use here.} From Lemma \ref{lemma_transfer_works} we conclude that all Clause Gadgets must have a node set to $1$. 
\end{proof}

These Lemmas lead us to the conclusion that indeed any satisfying assignment to the \textit{ONE-IN-THREE 3SAT} problem matches an NTPNE in the constructed graph, and vice versa, which directly implies that the NTPNE problem is NP-Hard. This result is formulated by the claim in Theorem \ref{theorem_SN_npc}, which we can now prove.
\begin{proof} \textit{(Theorem \ref{theorem_SN_npc})} Clearly the problem is in NP, as given any graph and a binary assignment to the nodes, we can verify that the assignment is an NTPNE in polynomial time. It is left to prove the problem is NP-Hard. Given a ONE-IN-THREE 3SAT instance, we construct a graph as described previously.
If there exists a satisfying assignment to the variables of the ONE-IN-THREE 3SAT instance, then we can assign $1$ to all Literal Nodes which represent variables that are assigned 'True', and to the necessary nodes within each gadget according to Lemma \ref{lemma_cg_pne}, and set all other nodes to $0$. Since we saw that the assignment described in Lemma \ref{lemma_cg_pne} forms an NTPNE, and that all connecting nodes (i.e. Transfer, Copy and Negation Nodes) do not affect this NTPNE given that they are all set to $0$, we will get an NTPNE.
In the other direction, if there exists an NTPNE in the constructed graph, From Lemmas \ref{lemma_literal_on}, \ref{lemma_one_literal}, \ref{lemma_transfer_works}, \ref{lemma_copy_works}, \ref{lemma_neg_works}, \ref{lemma_all_cg_on} we have that this NTPNE corresponds to a satisfying assignment to the ONE-IN-THREE 3SAT instance, where every Literal Node set to $1$ will translate to assigning 'True' to its matching variable, and every Literal Node set to $0$ will translate to assigning 'False' to its matching variable. 
\end{proof}
The proof itself gives a slightly more general result: The degree of the constructed graph is bounded\footnote{The reader who has read the details of the proof may verify that, in the version where the gadgets are chained as mentioned in footnotes \ref{chain_cg_footnote} and \ref{chain_copy_footnote}, a Literal Node is attached to 6 nodes within its Clause Gadget, and at most 2 Transfer Nodes, 2 Copy Nodes and 3 Negation Nodes.} by 13. Therefore, our proof does not require any information about the $15^{th}$ entry of the BRP onward, and thus extends to any NTPNE($T$) where $T$ agrees with the first 14 entries of the SN-BRP.
\begin{corollary}
Let $T$ be a BRP such that:
\begin{enumerate}
    \item T[1]=1
    \item $\forall k\in \{0,2,3,...,13\} \;\; T[k]=0$
\end{enumerate}
Then NTPNE($T$) is NP-complete.
\label{theorem_SN_corollary}
\end{corollary}
\section{Algorithm for the At-Most-Single-Neighbor Pattern}
\label{section_AMSN}
In this section, we focus on one of the most basic cases of a \textit{monotone} best response pattern, which is the \textit{At-Most-Single-Neighbor BRP}.
We remind the reader that in the At-Most-Single-Neighbor BRP, an agent prefers to produce the good iff at most one of their neighbors produces it. In addition, we formally define a \textit{monotone} best response pattern as follows:
\begin{definition}
A BRP $T$ is called monotonically increasing (resp. decreasing) if for all $k\in {\rm I\!N}$, $T[k]\leq T[k+1]$ (resp. $T[k]\geq T[k+1]$).
\end{definition}
The most basic (and well studied) case of a monotonically-decreasing BRP is the Best-Shot BRP. In ~\cite{IS_2007}, it is shown that in any PGG corresponding to the Best-Shot BRP, an NTPNE always exists. Thus, the decision problem in this case is trivially solvable in polynomial time. In this section, we show a similar result for the At-Most-Single-Neighbor BRP: we prove that an NTPNE exists in any PGG corresponding to this BRP, and present a polynomial time algorithm to find one. An important notion in Graph Theory, which will be of use during our proof, is the \textit{Maximum Independent Set}. Given a Graph $G=(V,E)$, an Independent Set (Henceforth IS) is a subset of nodes $S\subseteq V$ such that no two nodes in $S$ are adjacent. A \textit{Maximal Independent Set} is an IS $S$ such that for any node $v$ outside $S$ it holds that $S\cup\{v\}$ is not an IS. A \textit{Maximum Independent Set} $S$ is a Maximal IS, such that for any Maximal IS $S'$ it holds that $|S|\geq|S'|$. Using this notion, we constructively prove that an NTPNE of the At-Most-Single-Neighbor BRP exists in any graph, by providing a non-polynomial time algorithm to find one. We later alter the algorithm to work in polynomial time.
\begin{algorithm}
\hspace*{\algorithmicindent} \textbf{Input} graph $G=(V,E)$\\
\hspace*{\algorithmicindent} \textbf{Output} NTPNE for the At-Most-Single-Neighbor PGG on $G$
\begin{algorithmic}[1]
\caption{At-Most-Single-Neighbor: Non-Polynomial Time Algorithm}
\label{constructive_algo}
\State find a maximum IS $S\subseteq V$, and assign all its nodes with 1.
\State for each $v\in V\setminus S$, if $v$ has exactly one supporting neighbor at the moment of its assignment, then $v=1$. Otherwise $v=0$.
\end{algorithmic}
\end{algorithm}
\begin{theorem}
Let $T$ be the At-Most-Single-Neighbor BRP.
Given any graph $G$ as an input, Algorithm \ref{constructive_algo} outputs an NTPNE of the PGG defined on $G$ corresponding to $T$, and thus an NTPNE always exists. Therefore, the decision problem NTPNE($T$) is trivially solvable in polynomial time.
\label{Theorem_AMSN_DP}
\end{theorem}
\begin{proof}
A Maximum IS always exists in any graph, and therefore stage 1 of the algorithm, though not efficient, is well defined. Assume by way of contradiction that the assignment given by the algorithm is not a PNE. Then there must be some node $u$ that is not playing the best response to its neighbors assignments. Divide into two cases:

\textbf{Case 1}: If $u$ is playing $0$, then at the time of its assignment $u$ must have had at least two supporting neighbors, otherwise the algorithm would have assigned it with $1$. Since the algorithm never changes a node's assignment from 1 to 0, we have that also at the end of the run $u$ has at least two supporting neighbors, and therefore $u$ is playing its best response, in contradiction.

\textbf{Case 2}: If $u$ is playing $1$, then it must be that $u$ has at least two supporting neighbors, otherwise $u$ is playing its best response. Let $x,y$ be two supporting neighbors of $u$. Divide into two sub-cases:

\textbf{Sub-Case 2.1}: If $u\in S$, then $x,y\in V\setminus S$ (otherwise we have a contradiction to $S$ being an IS), and therefore $u$ had received its assignment before $x,y$ did. Any node not in $S$ is only assigned with $1$ by the algorithm if it has exactly one supporting neighbor at the time of the assignment, and therefore $x,y$ only had $u$ as a supporting neighbor at the time of the assignment. Specifically, $x,y$ are not adjacent to any other node in $S$, because all nodes in $S$ were already assigned $1$ by the time $x,y$ were assigned. Therefore, we have that $S':=(S\setminus\{u\})\cup\{x,y\}$ is also an IS. Since $|S'|>|S|$, we have a contradiction to the fact that $S$ is a Maximum IS.

\textbf{Sub-Case 2.2}: If $u\notin S$, then at least one of $x,y$ must have gotten its assignment after $u$, otherwise the algorithm would have assigned $u$ with $0$. w.l.o.g assume $x$ got its assignment after $u$, and in particular $x\in V\setminus S$. Since $x$ was assigned $1$, it had exactly one supporting neighbor at the time of the assignment, hence $u$ was its only supporting neighbor at the time of the assignment. Since $u\notin S$, we have that $x$ is not adjacent to any node in $S$, which means that $S\cup\{x\}$ is an IS, in contradiction to $S$ being a maximal IS. 

Therefore, all nodes play their best response, and so the assignment is an NTPNE.\footnote{Since $T[0]=1$, any PNE must be an NTPNE.}
\end{proof}

Theorem \ref{Theorem_AMSN_DP} shows that the decision problem is easy in this case, using the fact that a Maximum IS always exists in any graph. However, finding a Maximum IS is an NP-Hard problem, and so Algorithm \ref{constructive_algo} does not run in polynomial time. Nevertheless, it does provide a base for our following refined algorithm, which runs in polynomial time and finds an NTPNE in any given graph.
\begin{algorithm}
\hspace*{\algorithmicindent} \textbf{Input} graph $G=(V,E)$\\
\hspace*{\algorithmicindent} \textbf{Output} NTPNE for the At-Most-Single-Neighbor PGG on $G$
\begin{algorithmic}[1]
\caption{At-Most-Single-Neighbor: Polynomial Time Algorithm}
\State find a Maximal IS $S\subseteq V$.
\State perform stage 2 of Algorithm \ref{constructive_algo} using $S$
\If{the assignment is a PNE}
    \State return
\Else
    \State let $u$ be a node which isn't playing its best response
    \State $S \gets S \setminus \{u\}$
    \For{$x$ s.t $(x,u)\in E$}
        \If{$x$ is not adjacent to any node in $S$}
            \State $S \gets S\cup\{x\}$
        \EndIf
    \EndFor
\EndIf
\State go back to stage 2
\label{refined_poly_algo}
\end{algorithmic}
\end{algorithm}

\begin{theorem}
Let $T$ be the At-Most-Single-Neighbor BRP.
Given any graph $G$ as an input, Algorithm \ref{refined_poly_algo} runs in polynomial time and outputs an NTPNE of the PGG defined on $G$ corresponding to $T$.
\label{Theorem_AMSN_poly}
\end{theorem}

\begin{proof}
\textit{Correctness}:
We show that if at the beginning of iteration $S$ is a Maximal IS, then at the end of it $S$ increases by at least 1, and remains a Maximal IS. Therefore, after at most $|V|$ iterations, either the algorithm finds an NTPNE and stops, or $S$ increases enough to become a Maximum IS, and therefore by Theorem \ref{Theorem_AMSN_DP} the algorithm outputs an NTPNE.

\textbf{$S$ increases by at least 1 after each iteration:}
At the beginning of the iteration $S$ is a Maximal IS, but not necessarily a Maximum IS. Inspect the different cases of the proof of Theorem \ref{Theorem_AMSN_DP}, when assuming that the assignment is not a PNE. Notice that all the cases from the proof of Theorem \ref{Theorem_AMSN_DP} were contradicted by the fact that $S$ is a \textit{Maximal} IS (which is true for the current theorem as well), except for Sub-Case 2.1, which was contradicted by the fact that $S$ is a \textit{Maximum} IS. Therefore, if the assignment given by stage 2 of Algorithm \ref{constructive_algo} is not a PNE, the conditions of case 2.1 of the proof of Theorem \ref{Theorem_AMSN_DP} must hold. Thus, we are guaranteed that any node not playing its best response must be in $S$, and must have at least two neighbors which satisfy the condition of stage 9 in Algorithm \ref{refined_poly_algo}. Therefore, in each iteration of the Algorithm \ref{refined_poly_algo} $S$ gains at least 2 new nodes in the loop of stage 8, and loses exactly one node (which is $u$), and therefore the size of $S$ overall increases by at least 1. 

\textbf{$S$ remains a Maximal IS after each iteration:}
Since $S$ is a maximal IS at the beginning of the iteration, any node not in $S$ is adjacent to at least one node in $S$. In each iteration we only remove a single node $u$ from $S$, hence only $u$ and nodes that were adjacent to $u$ might (possibly) not be adjacent to any node in $S$. After iterating over all of $u$'s neighbors and adding whichever possible to $S$, we have that all of $u$'s neighbors are either in $S$ or adjacent to some node in $S$. Regarding $u$ itself, we have already shown that at least 2 of its neighbors must be added to $S$. Thus we have that $S$ remains a Maximal IS at the end of the iteration.
This concludes the correctness of the algorithm.

\textit{Run-Time}:
Stage 1 can be achieved in $O(|V|)$, greedily. Stages 2,3 and 8-10 all require iterating over all nodes, and for each node iterating over all its neighbors, i.e. $O(|V|^2)$. Therefore, each iteration of the algorithm runs in $O(|V|^2)$. As explained in the Correctness part of the proof, the number of iterations of the algorithm is bounded by $|V|$. Hence the overall run time of the algorithm is polynomial w.r.t the input. 
\end{proof}

\section{More Hard Patterns}
\label{section_more_hard_patterns}
Theorem \ref{theorem_SN_npc} provides a base to discover more classes of BRPs for which the decision problem is hard, which we present in this section. We focus specifically on non-monotone patterns, with a finite number of $1$'s, except for in Section \ref{add_1_0_Section} where we do not assume this. All proofs of this section are postponed to the appendix.
\begin{definition}
A BRP $T$ is called finite if it has a finite number of $1$'s, i.e.
\[\exists N\in {\rm I\!N} \;\;s.t\;\; \forall n>N \;\; T[n]=0\]
\end{definition}

\subsection{Flat Patterns}
\label{section_flat_patterns}
In this section we generalize the result of the Single-Neighbor BRP to any $t$-Neighbors BRP (where an agent's best response is $1$ iff exactly $t$ of their neighbors play 1), and in fact to an even more general case: we show that any BRP that is \textit{flat} (i.e. starting with $0$), non-monotone and finite, models an NP-complete decision problem.
\begin{definition}
A BRP $T$ is called \textit{flat}\footnote{Coined by Papadimitriou and Peng in ~\cite{Directed_Paper}.} if $T[0]=0$.
\label{def_flat}
\end{definition}
\begin{theorem}
Let $T$ be a BRP which satisfies the following conditions:
\begin{enumerate}
\item $T$ is flat
\item $T$ is non-monotone
\item $T$ is finite
\end{enumerate}
Then NTPNE($T$) is NP-complete.
\label{theorem_flat_NPC}
\end{theorem}
The proof can be found in Appendix \ref{appendix_flat_NPC}.
\subsection{Sloped Patterns}
\label{section_antennas}
In this section, we show that the decision problem remains hard when the assumption of flatness is replaced with the assumption that the BRP begins with a finite number of $1$'s, and at least two. That is, we prove hardness of any non-monotone, finite, \textit{sloped} BRP.
\begin{definition}
A BRP $T$ is called \textit{sloped} if $T[0]=T[1]=1$.
\end{definition}
\begin{theorem}
Let $T'$ be a BRP which satisfies the following conditions:
\begin{enumerate}
\item $T'$ is sloped
\item $T'$ is non-monotone
\item $T'$ is finite
\end{enumerate}
Then NTPNE($T'$) is NP-complete under Turing reduction.
\label{theorem_sloped_NPC}
\end{theorem}
The proof can be found in Appendix \ref{appendix_add_2_pluses_NPC}.
\subsection{Sharp Patterns Followed by Two $1$'s}
\label{section_+-++}
Among all non-monotone, finite BRPs, we have shown hardness of those which are flat, and those which are sloped. It is left to address the decision problem for non-monotone, finite BRPs which start with $1,0$.
\begin{definition}
A BRP $T$ is called \textit{sharp} if $T[0]=1$ and $T[1]=0$.
\end{definition}
In this section, we focus on patterns which start with $1$, followed by any positive, finite number of $0$'s, and then $1,1$. We prove that such BRPs present hard decision problems.
\begin{theorem}
Let $T'$ be a BRP which satisfies the following conditions:
\begin{enumerate}
\item $T'$ is finite
\item $T'$ is sharp
\item $\exists m\geq 2 \;\;$ s.t:
    \begin{enumerate}
        \item $\forall 1\leq k<m\;\; T'[k]=0$
        \item  $T'[m]=T'[m+1]=1$
    \end{enumerate}
\end{enumerate}
Then NTPNE($T'$) is NP-complete under Turing reduction.
\label{theorem_+-++}
\end{theorem}
The proof can be found in Appendix \ref{appendix_+-++}.
\subsection{Adding $1,0$ to Non-Flat Patterns}
\label{add_1_0_Section}
In this section, we show that any non-flat pattern for which the decision problem is hard remains hard when $1,0$ is added to the beginning of it. Notice that adding $1,0$ at the beginning of a non-flat pattern yields another non-flat pattern, and so, by using this result recursively we can add any finite number of $1,0$ to a non-flat hard pattern, and it will remain hard. So far, the only non-flat patterns we have shown to be hard are the ones from Theorems \ref{theorem_sloped_NPC} and \ref{theorem_+-++}. Notice that adding $1,0$ to a pattern of the form of \ref{theorem_sloped_NPC} simply gives a pattern of the form of \ref{theorem_+-++}, which we have already proved is hard. Thus, until other non-flat patterns are proved hard, the new class of patterns that are shown to be hard in this section is summarized by the form:
\[T=[\underbrace{1,0,1,0,1,0,....}_{finite\;number\;of\;'1,0'},\underbrace{1,0,0,0,...,1,1,?,?,...,0,0,0,...}_{pattern\;from\;Theorem\;\ref{theorem_+-++}}]\]
If other non-flat patterns are proved hard, this result could be applied to them as well.
We begin with a new definition:
\begin{definition}
Let $T,T'$ be two BRPs. We say that $T'$ is \textit{shifted by $m$} from $T$ if: \[\forall k\in {\rm I\!N} \;\; T'[k+m]=T[k]\]
We say that $T'$ is \textit{positively-shifted by $m$} from $T$ if in addition:
\[\forall k<m \;\; T'[k]=1\]
\label{def_shifted}
\end{definition}

\begin{theorem}
Let $T$ be a non-flat BRP s.t NTPNE($T$) is NP-complete. Let $T'$ be a BRP satisfying the following conditions:
\begin{enumerate}
    \item $T'$ is shifted by 2 from $T$.
    \item $T'$ is sharp
\end{enumerate} Then NTPNE($T'$) is NP-complete.
\label{theorem_add_+-_NPC}
\end{theorem}
The proof can be found in Appendix \ref{appendix_add_+-_NPC}.
\begin{subappendices}
\renewcommand{\thesection}{\Alph{section}}
\section{Proofs of Section \ref{section_more_hard_patterns}}
\subsection{Proof of Theorem \ref{theorem_flat_NPC}}
\label{appendix_flat_NPC}
\begin{proof}
Let $T$ be a BRP satisfying the conditions of Theorem \ref{theorem_flat_NPC}. Since $T$ is flat, $T[0]=0$, and since it is non-monotone, some entry of $T$ must be $1$. Since $T$ is finite, let $N$ be the largest entry of $T$ s.t $T[N]=1$, i.e. $\forall n>N \;\; T[n]=0$. If $N=1$, then $T$ is the Single-Neighbor-BRP, which we have already proved to be NP-Hard in Theorem \ref{theorem_SN_npc}. Otherwise, we reduce from NTPNE($SN-BRP$). Let $P$ be a PGG on a graph $G=(V,E)$ corresponding to the SN-BRP. Denote $V=\{v_1,...,v_n\}$. We construct the PGG $P'$ on the graph $G'=(V',E')$, corresponding to the BRP $T$. $G'$ is built from $N$ replicas of $G$, s.t each node in each replica is connected to all $N$ replicas of the neighbors it was originally connected to:
\[V'=\bigcup_{k \in [N]} \{v^{k}_1,...,v^{k}_n\}\]
\[E'=\{(v^{k}_i,v^{l}_j)|(v_i,v_j)\in E\}\]
We show that there exists an NTPNE in $P$ iff there exists an NTPNE in $P'$. Let $s=\{s_1,...,s_n\}$ be an NTPNE in $P$. We construct the strategy profile $s'$ of $P'$ such that:
\[\forall i\in[n] \;\;\forall k\in[N]\;\; s'^{k}_i=s_i\]
i.e. all replicas of each node get the assignment their original node had in $s$.
Let $v^{k}_i\in V'$ be some node corresponding to a node $v_i\in V$ of the original graph $G$. If $s'^{k}_i=1$, then $s_i=1$, and since $s$ is a PNE $v_i$ must have exactly one supporting neighbor in $G$. By the construction of $G'$, $v^{K}_i$ must have exactly $N$ supporting neighbors, since all replicas of each neighbor share the same assignment. Since $T[N]=1$, we have that $v^{k}_i$ is playing its best response. If $s'^{k}_i=0$, then $s_i=0$, and since $s$ is a PNE $v_i$ must have either 0 or at least two supporting neighbor in $G$. By the construction of $G'$, $v^{K}_i$ must have either 0 or at least $2N$ supporting neighbors, since all replicas of each neighbor share the same assignment. Since $T$ is flat, and $\forall k>N T[k]=0$ $v$, we have that $v^{k}_i$ is playing its best response. Therefore $s'$ is a PNE in $P'$. Furthermore, $s'$ is non-trivial because $s$ is non-trivial, thus $s'$ is an NTPNE.

In the other direction, let $s'$ be an NTPNE in $P'$. We arbitrarily choose replica number 1 of the original graph, and construct a strategy profile $s$ of $P$ as follows:
\[\forall i\in [n] \;\; s_i=s'^{1}_i\]
Notice that, in $G'$, all replicas of the same node share exactly the same neighbors. Therefore, their best response must be the same, and so they all must play the same strategy in any equilibrium, and so:
\[\forall i\in[n] \;\;\forall k,l\in[N]\;\; s'^{k}_i=s'^{l}_i\]
Specifically, there must be some non-zero assignment in $\{s'^{1}_i\}_{i\in[n]}$ (otherwise the entire strategy profile must be all-zeros, in contradiction), and therefore also in $\{s_i\}_{i\in[n]}$, so if $s$ is a PNE it is also an NTPNE.
In addition, by the construction of $G'$, any node $u'$ connected to some node $v'$ is connected to all $N$ replicas of $v'$, which all share the same assignment. Thus, by the construction of $s$, if a node $v_i\in V$ has $a$ supporting neighbors according to $s$, then $v^{1}_i\in V'$ must have $N\cdot a$ supporting neighbors according to $s'$. Specifically, the number of supporting neighbors each node in $G'$ has must be a multiple of $N$. Let $v_i\in V$ be some node in $G$. If $s_i=1$, then $s'^{1}_i=1$, and since $s'$ is a PNE $v'^{1}_i$ must have exactly $N$ supporting neighbors according to $s'$ (since the best response to any other multiple of $N$ is 0 according to $T$). Therefore, $v_i$ has exactly 1 supporting neighbor according to $s$, which means it is playing the best response according to the SN-BRP. If $s_i=0$, then $s'^{1}_i=0$, and since $s'$ is a PNE $v'^{1}_i$ must have either 0 or at least $2N$ supporting neighbors according to $s'$ (since the only other multiple of $N$ is $N$, which yields 1 as the best response according to $T$). Therefore, $v_i$ has either 0 or at least 2 supporting neighbors according to $s$, which means it is playing the best response according to the SN-BRP. 
\end{proof}

\subsection{Proof of Theorem \ref{theorem_sloped_NPC}}
\label{appendix_add_2_pluses_NPC}
\begin{proof}
For convenience, we demonstrate the general form of $T'$, where the marked indices will be defined shortly:
\[T'=[1,1,...,\underbrace{0}_{m},0,...,\underbrace{1}_{m'},?,?,...,0,0,...]\]
Let $m$ be the smallest index s.t $T'[m]=0$ (there exists one from condition 3). We define the BRP $T$ by:
\[\forall k\in {\rm I\!N} \;\; T[k]=T'[k+m]\]
i.e. $T'$ is positively-shifted by $m$ from $T$ (see Definition \ref{def_shifted}). Notice that $T$ satisfies the conditions of Theorem \ref{theorem_flat_NPC} (otherwise $T'$ must either be monotone or infinite, in contradiction). Hence, NTPNE($T$) is NP-Complete, and we can construct a Turing-reduction from it. Let $m'\geq1$ be the smallest index satisfying $T'[m']=1$ after the first $0$ in $T'$ ($m'$ must exist, since $T'$ is non-monotone). We have that $T'[m'-1]=0$. Let $P$ be a PGG on a graph $G=(V,E)$, where the BRP is $T$. Denote $V=\{v_1,...,v_n\}$. We construct $n$ PGGs $P'_1,...,P'_n$ corresponding to $T'$, on the graphs $G'_1=(V'_1,E'_1),...,G'_n=(V'_n,E'_n)$ respectively. For each $i\in [n]$, $G'_i$ contains the graph $G$, but in addition, each node except $v_i$ is connected to $m$ new 'Antenna' nodes that are only adjacent to that specific node. $v_i$ is connected only to $m-1$ such Antennas, and additionally to a \textit{Force-1-Gadget} (FG). Intuitively, in each graph $G'_i$ we add $m$ supporting neighbors (compared to $G$) to each node, and specifically force node $v_i$ to be assigned $1$ in any PNE.
\begin{figure}[h!]
\centering
\includegraphics[width=0.6\textwidth]{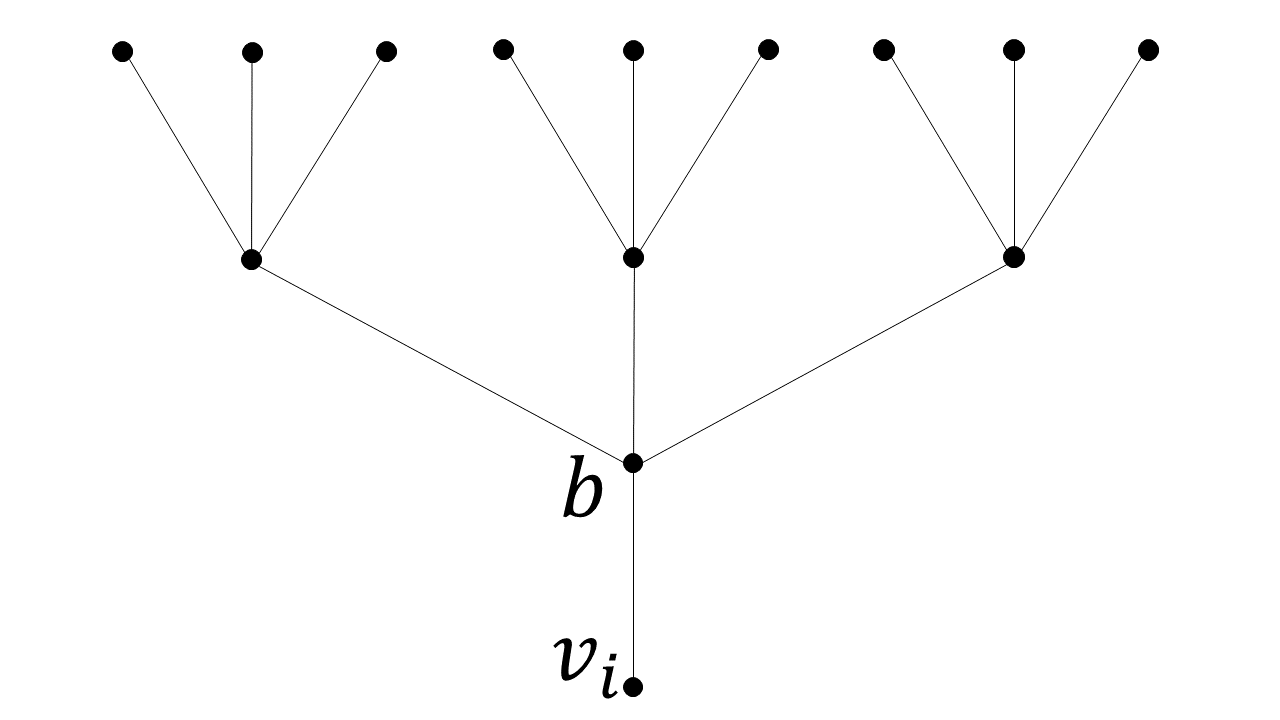}
\caption{\textit{Force-1-Gadget with $m'=4$}}
\label{fig_Antennas_FG}
\end{figure}

\textit{Force-1-Gadget} (shown in Figure \ref{fig_Antennas_FG}). The Force-1-Gadget of $v_i$ consists of 3 layers. The first layer consists of only a single 'bridge' node, denoted $b$, which is also the only node in the gadget connected to $v_i$. The second layer consists of $m'-1$ nodes, all of which are connected to $b$. The third layer consists of $(m'-1)^2$ nodes, such that each node of the previous layer is connected to $m'-1$ unique ones of them (i.e. layer 2 nodes don't share any layer 3 nodes).
\begin{lemma}
For any $i\in [n]$, in any PNE of $P'_i$, $v_i$ must be set to $1$. In addition, all nodes of the Force-1-Gadget must be set to $1$.
\label{lemma1_add_2_pluses}
\end{lemma}
\begin{proof}
We begin by showing the FG nodes must all play $1$.
All nodes in layer 3 of the FG must be set to 1, since they have only one neighbor each, and $T'[0]=T'[1]=1$.
Assume by way of contradiction that $b=0$. Then each node in layer 2 must play $0$, having exactly $m'-1$ supporting neighbors (the ones from layer 3). This leads to a contradiction in node $b$, which plays $0$ and yet has at most one supporting neighbor. Therefore, $b=1$.
Now, each node in layer 2 must also play $1$, having exactly $m'$ supporting neighbors (the ones from layer 3, and $b$). So indeed al nodes of the FG must play 1.

Now, assume by way of contradiction that $v_i=0$. Then node $b$ plays $1$ and yet has exactly $m'-1$ supporting neighbors, which is a contradiction. If $v_i=1$ on the other hand, there is no contradiction. So indeed $v_i$ must be set to $1$, and so do all nodes of the FG. 
\end{proof}

\begin{lemma}
For any $i\in [n]$, there exists an NTPNE in $P$ s.t $v_i=1$ iff there exists an NTPNE in $P'_i$.
\label{lemma2_add_2_pluses}
\end{lemma}
\begin{proof}
Fix $i\in[n]$, and let $s$ be an NTPNE in $P$ s.t $v_i=1$. We construct the strategy profile $s'$ of $P'_i$ such that 
\[\forall v_k \in V\;\; v'_k=v_k\]
\[\forall v' \in V'\setminus V \;\; v'=1\]
We now show that $s'$ is an NTPNE in $P'_i$.
Clearly, all Antenna nodes play their best response, having exactly one neighbor (since $T'[0]=T'[1]=1$). All FG nodes also play their best response, according to Lemma \ref{lemma1_add_2_pluses}. 

We now show the original nodes also play their best response. Let $v_k\in V$ be some node in $G'_i$. By definition $s'_k=s_k$. $v_k$ is adjacent in $G'_i$ to all nodes it was adjacent to in $G$, and by definition all its (original) neighbors receive the same assignment in $s'$ as they did in $s$. In addition, recall that $v_k$ has $m$ new neighbors, all of which receive an assignment of 1 (specifically, for $v_i$ one of those $m$ neighbors is the bridge node $b$ of the FG). Therefore, $v_k$ has $m$ more supporting neighbors in $G'_i$ according to $s'$ than it did in $G$ according to $s$. Since $T'$ is shifted by $m$ from $T$, and since $s$ is an NTPNE, we have that $v_k$ must be playing its best response in $s'$, and thus $s'$ is an NTPNE\footnote{Clearly $s'$ is non-trivial, for example the Antenna nodes are all assigned 1.}.

In the other direction, Let $s'$ be an NTPNE in $P'_i$. We construct the strategy profile $s$ of $P$ where
\[\forall v_k \in V\;\; s_k=s'_k\]
Let $v_k\in V$ be some node in $G$. In $G'_i$, the equivalent node had $m$ additional neighbors, all of which must be set to $1$ in $s'$. Therefore, by definition of $s$, $v_k$ has $m$ less supporting neighbors in $G_i$ according to $s$ than it did in $G'_i$ according to $s'$. Since $T'$ is shifted by $m$ from $T$, and since $s'$ is an NTPNE, we have that $v_k$ must be playing its best response in $s$, and thus $s$ is a PNE. In addition, from Lemma \ref{lemma1_add_2_pluses} we have that $v_i$ must play $1$ in $s'$, and therefore also in $s$, hence $s$ is non-trivial, i.e. $s$ is an NTPNE. 
\end{proof}
We now continue with the proof of the theorem. According to Lemma \ref{lemma2_add_2_pluses}, if there doesn't exist an NTPNE in $P$, then for all $i\in[n]$ there doesn't exist an NTPNE in $P'_i$.
On the other hand, if there does exist an NTPNE $s$ in $P$, then there must be some node $v_i\in V$ s.t $v_i=1$ according to $s$. Therefore, from Lemma \ref{lemma2_add_2_pluses}, we have that $P'_i$ has an NTPNE. Therefore, given an oracle $A$ which solves NTPNE($T'$), we run $A$ on each of the games $P'_1,...,P'_n$. If there exists an NTPNE in one of them, there must exist one in $P$, and otherwise there must not exist one in $P$. 
\end{proof}

\subsection{Proof of Theorem \ref{theorem_+-++}}
\label{appendix_+-++}
\begin{proof}
We define the BRP $T$ by:
\[\forall k\in {\rm I\!N} \;\; T[k]=T'[k+1]\]
i.e. $T'$ is positively-shifted by 1 from $T$. Notice that $T$ satisfies all conditions of Theorem \ref{theorem_flat_NPC} (specifically, the conditions of $T'$ imply that $T$ is non-monotone). Hence, NTPNE($T$) is NPC, which allows us to construct a Turing-reduction from it. Let $N$ be the largest index s.t $T'[N]=1$, i.e. $\forall k>N \;\; T'[k]=0$.
For convenience, we demonstrate the general form of $T'$:
\[T'=[1,0,0,...,\underbrace{1}_{m},1,?,?,...,\underbrace{1}_{N},0,0,...]\]
where $m$ is defined in Theorem \ref{theorem_+-++} itself. Let $P$ be a PGG defined by $T$ on a graph $G=(V,E)$. Denote $V=\{v_1,...,v_n\}$. We construct $n$ PGGs $P'_1,...,P'_n$ corresponding to $T'$, on the graphs $G'_1=(V'_1,E'_1),...,G'_n=(V'_n,E'_n)$ respectively. For each $i\in [n]$, $G'_i$ contains the graph $G$, and in addition, each node is connected to a unique \textit{Node Gadget} (NG) composed of $2m$ nodes, and node $v_i$ is additionally connected to a \textit{Force-1-Gadget} (FG).
\begin{figure}[h!]
\begin{minipage}{.5\textwidth}
    \centering
    \includegraphics[width=\textwidth]{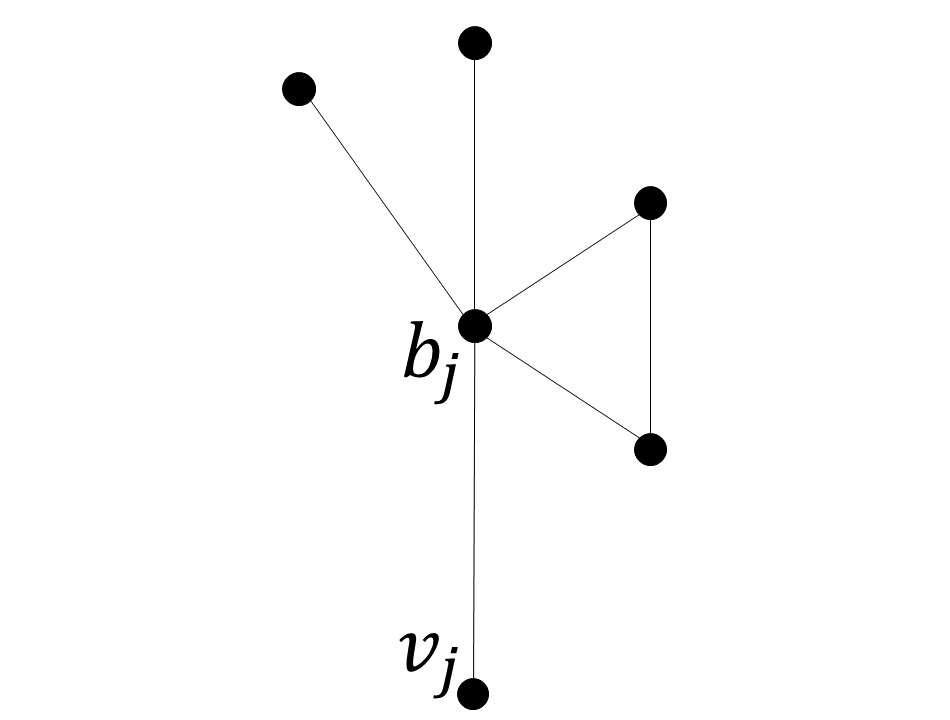}
    \caption{\textit{Node-Gadget with $m=3$}}
    \label{fig_+-++_NG}
\end{minipage}%
\begin{minipage}{.5\textwidth}
    \centering
    \includegraphics[width=\textwidth]{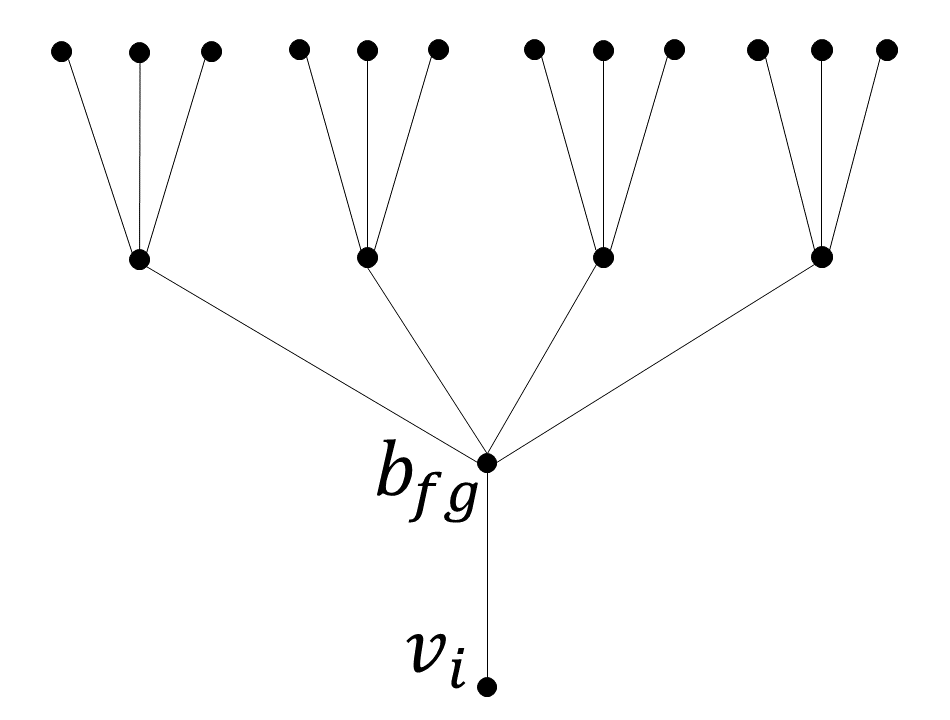}
    \caption{\textit{Force-1-Gadget with $m=3$, $N=4$}}
    \label{fig_+-++_FG}
\end{minipage}
\end{figure}

\textit{Node Gadget} (shown in Figure \ref{fig_+-++_NG}). For each node $v_j\in V$, the Node Gadget is defined as follows. We construct a $(m+1)-Clique$ Denoted $C_j$, from which one of the nodes, denoted $b_j$, will be referred to as the 'ng-bridge' (node-gadget-bridge) node. To $b_j$, we connect $m-1$ 'Antenna' nodes, all of which have only $b_j$ as a neighbor. We also connect $b_j$ to $v_j$.

\textit{Force-1-Gadget} (shown in Figure \ref{fig_+-++_FG}). The Force-1-Gadget of $v_i$ consists of 3 layers. The first layer consists of only a single 'fg-bridge' (force-gadget-bridge) node, denoted $b_{fg}$, which is also the only node in the gadget connected to $v_i$. The second layer consists of $N$ nodes, all of which are connected to $b_{fg}$. The third layer consists of $m\cdot N$ nodes, such that each node of the previous layer is connected to $m$ unique ones of them (i.e. layer 2 nodes don't share any layer 3 nodes).

We begin by proving the following lemmas.
\begin{lemma}
Fix $i\in [n]$, and let $s'$ be a PNE of the PGG $P'_i$. Let $v_j\in V$ be some node in $G'_i$. Then the ng-bridge node $b_j$ must be set to $1$ according to $s'$. In addition, when $b_j=1$ there exists an assignment s.t all nodes of the Node-Gadget play their best response.
\label{lemma1_+-++}
\end{lemma}
\begin{proof}
Assume by way of contradiction that $b_j=0$. Since the 'Antennas' of the node gadget are only adjacent to $b_j$, their best response according to $T'$ is $1$ regardless of $b_j$'s strategy, hence they all must be set to $1$. Therefore $b_j$ has (at least) $m-1$ supporting neighbors. Examine the remaining nodes of the Clique $C_j$ of the NG (excluding $b_j$). Since $b_j=0$, they can't all be set to $0$, because then each of them is not playing its best response. Therefore, at least one node in $C_j$ must be playing $1$, denoted $c$. Hence, all other nodes in $C_j$ must be set to $0$, otherwise $c$ would have between $2$ to $m-1$ supporting neighbors, thus it would not be playing its best response according to $T'$. Altogether, $b_j$ must have either $m$ or $m+1$ supporting neighbor, depending on the strategy of $v_j$. In both cases $b_j$ is not playing its best response according to $T'$, in contradiction to $s'$ being a PNE.

Note that when $b_j=1$, we can set all its Antenna nodes to $0$ and all of the nodes in the Clique $C_j$ to $1$, and have that all nodes in the NG play their best response (regardless of the strategy of $v_j$). 
\end{proof}

\begin{lemma}
For any $i\in [n]$, in any PNE of $P'_i$, $v_i$ must be set to $1$. In addition, there is only one possibility for an assignment of the Force-1-Gadget nodes, and in this assignment $b_{fg}=0$.
\label{lemma2_+-++}
\end{lemma}
\begin{proof}
We begin by showing the assignment of the FG nodes.
Assume by way of contradiction that the fg-bridge node $b_{fg}$ is set to $1$. Let $x$ be some node from layer 2 of the FG. If $x=1$ then all its Antenna nodes must be set to $1$ according to $T'$, and therefore $x$ has one supporting neighbor, and prefers playing $0$. If $x=0$, then all its Antenna nodes must be set to $0$ according to $T'$, and therefore $x$ has $m+1$ supporting neighbors, and prefers playing $1$. Hence, $b_fg=0$. Now, if $x$ is playing $0$, all its Antenna nodes must play $1$, hence $x$ has $m$ supporting neighbors, which means it isn't playing its best response. Thus, all nodes in layer 2 of the FG must play $1$, and therefore all Antenna nodes (layer 3 of the FG) must play $0$.

Now, assume by way of contradiction that $v_i=0$. Then the fg-bridge node $b_{fg}$ has $N$ supporting neighbors (from layer 2), and yet plays $0$, in contradiction to $T'$. Therefore, $v_i$ must be set to 1. In this case, $b_{fg}$ has $N+1$ supporting neighbors, which means it is indeed playing its best response, concluding the proof of the lemma. 
\end{proof}

\begin{lemma}
Fix $i\in [n]$. There exists an NTPNE in $P$ s.t $v_i=1$ iff there exists an NTPNE in $P'_i$.
\label{lemma3_+-++}
\end{lemma}
\begin{proof}
Let $s$ be an NTPNE in $P$ s.t $v_i=1$ according to $s$. We construct the strategy profile $s'$ of $P'_i$ such that all NG-nodes and FG-nodes are assigned according to the PNEs described in the proofs of Lemmas \ref{lemma1_+-++} and \ref{lemma2_+-++} respectively, and additionally
\[\forall v_k \in V\;\; v'_k=v_k\]
We now show that $s'$ is an NTPNE in $P'_i$.
We have already shown in Lemmas \ref{lemma1_+-++} and \ref{lemma2_+-++} that all NG nodes and FG nodes play their best response. It remains to be shown that the original nodes also play their best response. Let $v_k\in V$ be some node in $G'_i$. By definition $s'_k=s_k$. $v_k$ is adjacent in $G'_i$ to all nodes it was adjacent to in $G$, and by definition all its (original) neighbors receive the same assignment in $s'$ as they did in $s$. In addition, from Lemma \ref{lemma1_+-++} we have that $v_k$ has one additional supporting neighbor, which is the ng-bridge node $b_k$. From Lemma \ref{lemma2_+-++} we have that even if $k=i$, the additional fg-bridge node is assigned 0, so it does not affect $v_k$. Therefore, $v_k$ has exactly one additional supporting neighbor in $P'_i$ than it does in $P$. Since $T'$ is shifted by 1 from $T$, and since $s$ is a PNE, we have that $v_k$ must be playing its best response in $s'_i$ as well, and thus $s'$ is an NTPNE\footnote{clearly there exists a non-zero assignment in $s'$, for example the ng-bridge nodes.}.

In the other direction, Let $s'$ be an NTPNE in $P'_i$. We construct the strategy profile $s$ of $P$ where
\[\forall v_k \in V\;\; s_k=s'_k\]
Let $v_k\in V$ be some node in $G$. In $G'_i$, the equivalent node had exactly the same neighbors from $V$ (all of which play the same in both games), and according to Lemmas \ref{lemma1_+-++} and \ref{lemma2_+-++}, it had exactly one additional supporting neighbor from the node gadget. Therefore, by definition of $s$, $v_k$ has 1 less supporting neighbor in $G_i$ according to $s$ than it did in $G'_i$ according to $s'$. Since $T'$ is shifted by 1 from $T$, and since $s'$ is an NTPNE, we have that $v_k$ must be playing its best response in $s$, and thus $s$ is a PNE. In addition, from Lemma \ref{lemma2_+-++} we have that $v_i$ must play $1$ in $s'$, and therefore also in $s$, hence $s$ is non-trivial, i.e. $s$ is an NTPNE. 
\end{proof}

We now continue with the proof of the theorem. According to Lemma \ref{lemma3_+-++}, if there doesn't exist an NTPNE in $P$, then for all $i\in[n]$ there doesn't exist an NTPNE in $P'_i$.
On the other hand, if there does exist an NTPNE $s$ in $P$, then there must be some node $v_i\in V$ s.t $v_i=1$ according to $s$. Therefore, from Lemma \ref{lemma3_+-++}, we have that $P'_i$ has an NTPNE. Therefore, given an oracle $A$ which solves NTPNE($T'$), we run $A$ on each of the games $P'_1,...,P'_n$. If there exists an NTPNE in one of them, there must exist one in $P$, and otherwise there must not exist one in $P$. 
\end{proof}

\subsection{Proof of Theorem \ref{theorem_add_+-_NPC}}
\label{appendix_add_+-_NPC}
\begin{proof}
We construct a many-one reduction from NTPNE($T$). Let $P$ be a PGG defined by $T$ on a graph $G=(V,E)$. Denote $V=\{v_1,...,v_n\}$. We construct the PGG $P'$ corresponding to $T'$, on the graph $G'=(V',E')$. $G'$ contains the graph $G$, and in addition, each node is connected to a unique \textit{Node Gadget} (NG) composed of 4 nodes.
\begin{figure}[h!]
\centering
\includegraphics[width=0.6\textwidth]{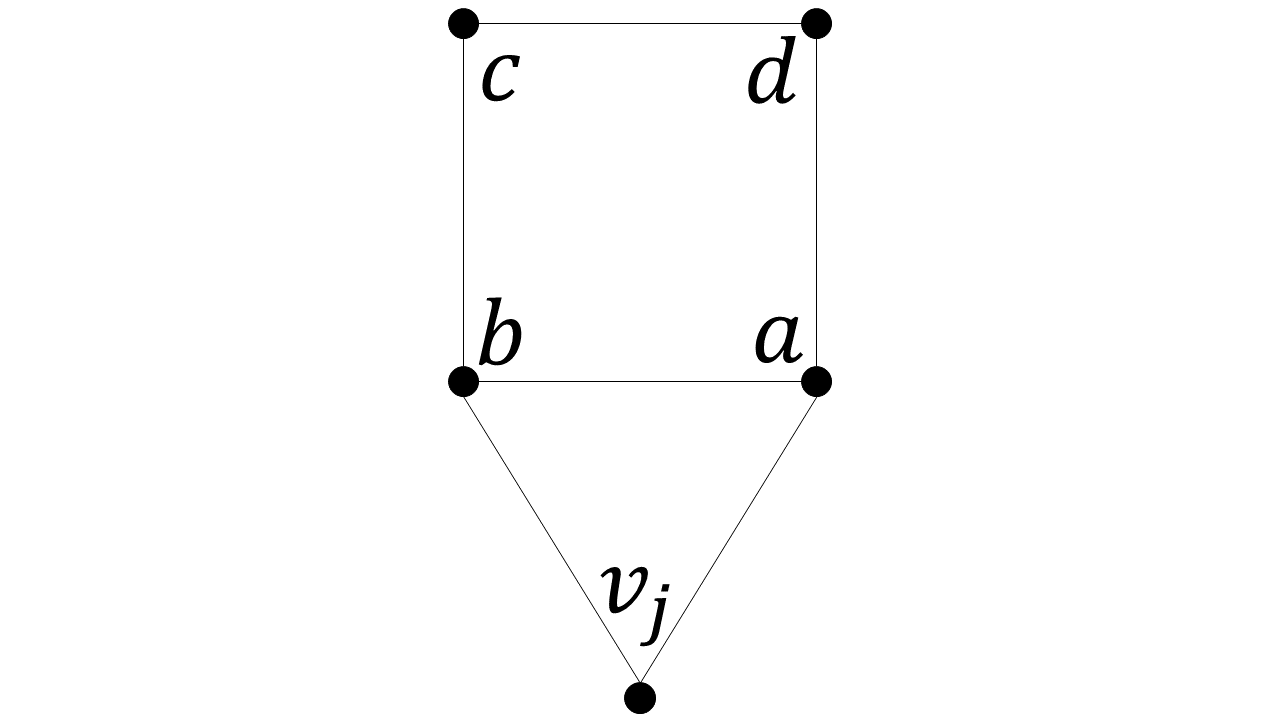}
\caption{\textit{Node-Gadget}}
\label{fig_+-+_NG}
\end{figure}

\textit{Node-Gadget} (shown in Figure \ref{fig_+-+_NG}). For each node $v_j\in V$, the 4 nodes of its Node Gadget, denoted $a,b,c,d$, form a cycle. $v_j$ is connected to $a,b$.

\begin{lemma}
Let $s'$ be a PNE of the PGG $P'$, and let $v_j\in V$ be some node in $G'$. Then nodes $a,b$ of $v_j$'s Node-Gadget must be set to $1$ according to $s'$. In addition, if $a=b=1$ there exists an assignment to $c,d$ s.t $a,b,c,d$ play their best response.
\label{lemma_+-+}
\end{lemma}
\begin{proof}
We first prove $a,b$ must be set to $1$. Assume by way of contradiction that the claim is incorrect. Divide into two cases.

\textbf{Case 1}: Assume $v_j$ is playing $0$. If both $a$ and $b$ play $0$ then each of them must have some supporting neighbor, and therefore $c, d$ must play $1$. This leads to a contradiction in $c, d$, each of which has one supporting neighbor and, according to $T'$, prefers playing $0$. If only one of $a, b$ plays $0$, w.l.o.g $a=1$ and $b=0$, then $c$ must play $0$, otherwise $b$ has two supporting neighbors and prefers to play $1$. $d$ has only one supporting neighbor and thus, by the definition of $T'$, prefers to play $0$. The contradiction comes from $c$ which has no supporting neighbors, and yet plays 0. 
Notice that, in this case, if $a=b=c=d=1$ then they all play their best response.

\textbf{Case 2}: Assume $v_j$ is playing $1$. If both $a$ and $b$ play $0$ then $c,d$ must play $0$, otherwise $a$ or $b$ would have two supporting neighbors and would prefer playing $1$ according to $T'$. This leads to a contradiction in $c,d$, both of which have no supporting neighbors and according to $T'$ prefer playing $1$. If only one of $a,b$ plays $0$, w.l.o.g $a=1$ and $b=0$, then $c$ must play $1$, otherwise $b$ has two supporting neighbors and prefers to play $1$. $d$ has two supporting neighbors and thus, by the definition of $T'$, prefers to play 1. The contradiction comes from $c$ which has one supporting neighbors, and yet plays 1.
Notice that, in this case, if $a=b=1,c=d=0$ then they all play their best response. 
\end{proof}

We now continue with the proof of the theorem, showing that there exists an NTPNE in $P$ iff there exists one in $P'$. 
Let $s$ be an NTPNE in $P$. We construct the strategy profile $s'$ of $P'$ such that 
\[\forall v_j \in V\;\; s'_j=s_j\]
and additionally, all nodes of the node gadgets are set according to the assignments shown in Lemma \ref{lemma_+-+}, i.e. for each $j\in[n]$ if $v_j=0$ then $a=b=c=d=1$, and if $v_j=1$ then $a=b=1,c=d=0$.

As stated during the proof of Lemma \ref{lemma_+-+}, all of the Node-Gadget nodes indeed play their best response in $s'$. It is left to show the same for the original nodes. Let $v_j\in V$ be some node in $G'$. By definition $s'_j=s_j$. $v_j$ is adjacent in $G'$ to all nodes it was adjacent to in $G$, and by definition all its (original) neighbors receive the same assignment in $s'$ as they did in $s$. In addition, according to Lemma \ref{lemma_+-+}, $v_j$ must have two additional supporting neighbors from the Node-Gadget. Therefore, if $v_j$ had $k$ supporting neighbors in $G$ according to $s$, it now has $k+2$ supporting neighbors in $G'$ according to $s'$. Since $T'$ is shifted by 2 from $T$, and since $s$ is an NTPNE, we have that $v_j$ must be playing its best response in $s'$, and thus $s'$ is a PNE. Moreover, $s'$ is clearly non-trivial\footnote{For example, nodes $a,b$ of each NG must be set to $1$.}, and thus it is an NTPNE.

In the other direction, Let $s'$ be an NTPNE in $P'$. We construct the strategy profile $s$ of $P$ such that:
\[\forall v_j \in V\;\; s_j=s'_j\]
Let $v_j\in V$ be some node in $G$. According to Lemma \ref{lemma_+-+}, the equivalent node in $G'$ had 2 additional supporting neighbors from its Node-Gadget. Other than these neighbors, $v_j$'s neighbors are the same in $G'$ and in $G$, and by definition of $s$ they all play the same strategy in both games. Therefore, $v_j$ has exactly 2 less supporting neighbors in $G$ according to $s$ than it did in $G'$ according to $s'$. Since $T'$ is shifted by 2 from $T$, and since $s'$ is an NTPNE, we have that $v_j$ must be playing its best response, and thus $s$ is a PNE. It is left to show $s$ is not trivial. Since $T$ is non-flat, and $T'$ is shifted by 2 from $T$, we have that $T'[2]=1$. Assume by way of contradiction that all $v_j\in V$ play $0$ according to $s'$, i.e. only the NG nodes $a,b$ of each node are assigned $1$. Then we have that each $v_j\in V$ has exactly 2 supporting neighbors, thus not playing its best response according to $T'$, in contradiction to $s'$ being a PNE. Therefore, there must be some node $v_j\in V$ s.t $v_j=1$ in $s'$, and by definition of $s$, $v_j=1$ in $s$ as well. So, $s$ is non-trivial, i.e. $s$ is an NTPNE. 
\end{proof}
\end{subappendices}
\bibliographystyle{splncs04}
\bibliography{mybibliography}

\begin{thebibliography}{8}

\bibitem{IS_2007}
Bramoullé Y, Kranton R.: Public goods in networks. Journal of Economic theory \textbf{135}(1), 478-494 (2007)

\bibitem{Modifications}
Kempe, D., Yu, S., Vorobeychik, Y.: Inducing equilibria in networked public goods games through network structure modification. arXiv preprint arXiv:2002.10627 (2020)

\bibitem{Parameterized}
Maiti, A., Dey, P.: On parameterized complexity of binary networked public goods game. arXiv preprint arXiv:2012.01880 (2020)

\bibitem{Directed_Paper}
Papadimitriou, C., Peng, B.: Public goods games in directed networks. In: Proceedings of the 22nd ACM Conference on Economics and Computation 2021, pp. 745--762. (2021)

\bibitem{one_in_three_3sat}
Schaefer, Thomas J.: The complexity of satisfiability problems. In: Proceedings of the tenth annual ACM symposium on Theory of computing, p. 216-226. (1978)

\bibitem{Point_Out_Flaw}
Yang, Y., Wang, J.: A refined study of the complexity of binary networked public goods games. arXiv preprint arXiv:2012.02916, (2020)

\bibitem{Flawed}
Yu, S., Zhou, K., Brantingham, J., Vorobeychik, Y.: Computing equilibria in binary networked public goods games. In Proceedings of the AAAI Conference on Artificial Intelligence, Vol. 34, No. 02, pp. 2310--2317. (2020)

\bibitem{Corrected}
Yu, S., Zhou, K., Brantingham, J., Vorobeychik, Y.: Computing equilibria in binary networked public goods games. arXiv:1911.05788v3 (2021)

\end{thebibliography}

\end{document}